%% file: main.tex
\documentclass[12pt]{article}
\input{preamble-main}

\title{\textbf{Principal component analysis in econometrics:\\a selective inference perspective}%
\thanks{The authors would like to thank Yoshihiko Nishiyama, Takahide Yanagi, Yuta Okamoto and Shunsuke Imai for helpful comments and discussions. All errors are our own. An earlier version of this paper circulated under the title ``Estimating the true number of principal components under the random design.''}}

\author{
    Yasuyuki Matsumura\thanks{Graduate School of Economics, Kyoto University. Email: \href{mailto:matsumura.yasuyuki.w85@kyoto-u.jp}{matsumura.yasuyuki.w85@kyoto-u.jp}.}
    \and
    Chisato Tachibana\thanks{Graduate School of Economics, University of Tokyo. Email: \href{mailto:tachibana-chisato215@g.ecc.u-tokyo.ac.jp}{tachibana-chisato215@g.ecc.u-tokyo.ac.jp}.}
}

\begin{document}
\maketitle

%%%%%%%%%%%%%%%%%%%%%%%%%%%%%%%%%%%%%%%%%%%%%%%%%%%%%%%%%%%%%%%%%%%%%%%
\begin{abstract}
We study the long-standing problem of determining the number of principal components in econometric applications from a selective inference perspective. We consider i.i.d.\ observations from a $p$-dimensional random vector with $p<n$ and define the ``true'' dimensionality as the rank of the population covariance matrix. Building on the sequential testing viewpoint, we propose a data-driven procedure that estimates $\rank(\Sigma_X)$ using a statistic that depends on the eigenvalues of the sample covariance matrix. While the test statistic shares the functional form of its fixed design counterpart \cite{Choi-Taylor-Tibshirani-2017}, our analysis departs from the non-stochastic setting by treating the design as random and by avoiding parametric Gaussian assumptions. Under a locally defined null hypothesis, we establish asymptotically exact type~I error controls in the sequential testing procedure, with simulation results indicating empirical validity of the proposed method.
\end{abstract}

%%%%%%%%%%%%%%%%%%%%%%%%%%%%%%%%%%%%%%%%%%%%%%%%%%%%%%%%%%%%%%%%%%%%%%%
\section{Introduction}\label{sec:intro}

Datasets with a large number of covariates are now commonplace across the economics and social sciences. A recurring challenge in such situations is multicollinearity: when predictors are strongly correlated, standard inferential and predictive procedures can be unstable and difficult to interpret. Exploratory data analysis (EDA) therefore plays an important role in revealing lower-dimensional structure and guiding subsequent modeling decisions. EDA spans a spectrum from informal, trial-and-error practices (sometimes shading into ``$p$-hacking'') to formal, criterion-based approaches. Examples of the latter include model selection via information criteria \citep{Akaike-1973,Schwarz-1978} and regularization methods such as the lasso \citep{Tibshirani-1996}.

In this paper we focus on principal component analysis (PCA), a classical tool for dimension reduction dating back to \citet{Pearson-1901} and \citet{Hotelling-1933}. PCA summarizes the dominant directions of variability in a multivariate distribution and often serves as a first step for both interpretation and downstream tasks such as regression, clustering, or forecasting. Beyond its methodological elegance, PCA is widely used in practice. In development economics, for instance, researchers routinely construct asset-based wealth indices using PCA when direct income or expenditure data are unavailable; see, among many others, \citet{Filmer-Pritchett-2001}, \citet{McKenzie-2005}, and \citet{Mbiti-Muralidharan-Romero-Schipper-Manda-Rajani-2019}. The original covariates in these applications typically comprise housing characteristics, access to infrastructure, and ownership of durable goods, whereas the leading principal component is interpreted as a wealth index or a measure of household economic status. In macroeconomic forecasting, factor-augmented models exploit PCA to condense information from large panels (often 100-200 time series spanning prices, consumption, production, employment, and other indicators) with the resulting principal components improving predictive performance; see, for example, \citet{Stock-Watson-2002b} and \citet{Bernanke-Boivin-Eliasz-2005}.

Despite its ubiquity, a central question in PCA remains unsettled: how many principal components should be retained? Numerous rules-of-thumb and formal procedures have been proposed, but no community-wide standard has emerged. Researcher-dependent practices include selecting the smallest number of components explaining a pre-specified fraction (often 70--90\%) of total variation; applying Kaiser's rule to retain components with variance greater than one \citep{Kaiser-1960}; or visually inspecting a scree plot \citep{Cattell-1966}. Formalized methods include distribution-based tests such as \citet{Bartlett-1950}, information criteria \citep{Bai-Ng-2003}, and eigenvalue-based methods \citep{Anderson-1963,Waternaux-1976,Muirhead-1982,Tyler-1983,Johnstone-2001,Kleibergen-Paap-2006,Kritchman-Nadler-2009,Hallin-Paindaveine-Verdebout-2010,Onatski-2010,Ahn-Horenstein-2013}. A comprehensive summary is provided, for example, in \citet[Chapter~6]{Jolliffe-2002}.

In this paper, we consider an i.i.d.\ sample $\{X_i\}_{i=1}^n$ from a $p$-dimensional random vector $X\in\mathbb{R}^p$ with $p<n$ for some fixed $p$. We take the true number of underlying principal components to be the rank of the population covariance matrix $\Sigma_X=\E[X_i X_i^{\top}]$. Using a test statistic that depends on the eigenvalues of the sample covariance matrix
$\hat{\Sigma}_X = n^{-1} \sum_{i=1}^n X_i X_i^{\top}$,
we estimate $\rank(\Sigma_X)$ via a sequential hypothesis testing procedure. Without imposing Gaussian assumptions on $X$, we establish that the resulting procedure achieves type~I error controls asymptotically under a locally defined null hypothesis.

%%%%%%%%%%%%%%%%%%%%%%%%%%%%%%%%%%%%%%%%%%%%%%%%%%%%%%%%%%%%%%%%%%%%%%%
\subsection*{Plan of the article}

The remainder of this section reviews related literature. Section~\ref{sec:theory} introduces the setup and assumptions and formalizes our testing procedure. Section~\ref{sec:simulation} presents a simulation study that illustrates the performance of the method. 
%Section~\ref{sec:empirical} demonstrates the practical usefulness of our approach in several empirical applications. 
Section~\ref{sec:conclusions} concludes.

%%%%%%%%%%%%%%%%%%%%%%%%%%%%%%%%%%%%%%%%%%%%%%%%%%%%%%%%%%%%%%%%%%%%%%%
\subsection*{Related literature}

Our approach connects to several strands of work and, conceptually, aligns with the sequential testing procedure for choosing the number of principal components under the fixed design \citep{Choi-Taylor-Tibshirani-2017}. While our test statistic in \eqref{eq:random-csv-stat} shares the same functional form as Test 2.1 in \citet{Choi-Taylor-Tibshirani-2017}, our theoretical analysis departs from their setting: whereas they assume a non-stochastic design and Gaussian errors, we treat the observed data as realizations of random variables and avoid Gaussian assumptions.

Our method is also aligned with frameworks for sequential hypothesis testing \citep{GSell-Wager-Chouldechova-Tibshirani-2016}. Methodologically, it relates to the ``Kac-Rice'' test theory \citep{Taylor-Loftus-Tibshirani-2016}, and the viewpoint of selective conditional inference \citep{Lee-Sun-Sun-Taylor-2016,Tibshirani-Taylor-Lockhart-Tibshirani-2016}. Our analysis further draws on random matrix theory and matrix perturbation theory \citep{Bhatia-1997,Horn-Johnson-2013,Yu-Wang-Samworth-2015} and on ideas from the theory of statistical experiments \citep[Chapters~6-9]{van-der-Vaart-1998}.

A closely related dimension-reduction framework is factor analysis (FA) \citep{Spearman-1904,Thurstone-1947}. A large literature studies data-driven selection of the number of factors, especially in macroeconomic and financial panels; see, for example, \citet{Stock-Watson-2002a}, \citet{Bai-Ng-2003,Bai-Ng-2007}, \citet{Hallin-Liska-2007}, \citet{Kritchman-Nadler-2008}, \citet{Onatski-2010}, \citet{Ahn-Horenstein-2013}, \citet{Trapani-2018}, and \citet{Fan-Ke-Sun-Zhou-2019}. 

Finally, when the dimension of covariates at hand exceeds the sample size ($p>n$), sparse variants of PCA and FA have been developed to ensure identifiability and interpretability; see \citet{Zou-Hastie-Tibshirani-2006}, \citet{Bai-Li-2012}, \citet{Caner-Han-2014}, \citet{Fan-Ke-Wang-2020}, \citet{Qiu-Li-Yao-2024}, and \citet{Li-Zhang-2025}.

%%%%%%%%%%%%%%%%%%%%%%%%%%%%%%%%%%%%%%%%%%%%%%%%%%%%%%%%%%%%%%%%%%%%%%%
\section{Theoretical analysis}\label{sec:theory}
Suppose that we observe an i.i.d.\ sample $\{X_i\}_{i=1}^n$ from the random vector $ X \in \mathbb{R}^p$ in the population. Here, $p$ is a given constant and we assume $p < n$. For simplicity and without loss of generality, we assume that $\mathbb{E}[X_i] = 0$. We denote
\[
\Sigma_X = \mathbb{E}[X_i X_i^{\top}], \qquad 
\hat{\Sigma}_X = \frac{1}{n} \sum_{i=1}^n X_i X_i^{\top}.
\]
Furthermore, let $\lambda_1 \geq \cdots \geq \lambda_p$ be the population eigenvalues of $\Sigma_X$ sorted in descending order, and let $\hat{\lambda}_1 \geq \cdots \geq \hat{\lambda}_p$ be the sample eigenvalues of $\hat{\Sigma}_X$ sorted in descending order.
Our target is
\[
k_0 = \operatorname{rank}(\Sigma_X).
\]
To estimate $k_0$, the analyst sequentially conducts tests of the following nested hypotheses: for each step $k = 1, \ldots, p-1$,
\begin{align}\label{eq:null-and-alter-hypo}
\H_{0,k} : \operatorname{rank}(\Sigma_X) \leq k-1
\quad \text{vs} \quad
\H_{1,k} : \operatorname{rank}(\Sigma_X) \geq k.
\end{align}
For testing \eqref{eq:null-and-alter-hypo}, we propose the following sequential testing procedure:

\begin{test}\label{test:proposed-testing-procedure}
Define the test statistic as follows:
\begin{align}
                        \hat{\S}_{k,0} = \frac{
                              \int_{\hat{\lambda}_k}^{\hat{\lambda}_{k-1}} \exp\left({-\frac{nu}{2 \hat{\sigma}_k^2}}\right) u^{\frac{n-p-1}{2}} \prod_{j \neq k}^{p} | u - \hat{\lambda}_j | du
                        }{
                              \int_{\hat{\lambda}_{k+1}}^{\hat{\lambda}_{k-1}} \exp\left({-\frac{nu}{2 \hat{\sigma}_k^2}}\right) u^{\frac{n-p-1}{2}}  \prod_{j \neq k}^{p} | u - \hat{\lambda}_j | du
                        }, 
                        \quad 
                        \hat{\lambda}_0 = \infty,
                        \quad
                        \hat{\sigma}_k^2=\frac{(1+\hat{\kappa})\sum_{j=k}^{p}\hat{\lambda}_j}{p-k+1},
                        \label{eq:random-csv-stat}
                  \end{align}   
where $\hat{\kappa}$ is a consistent estimator of the kurtosis\footnote{See, e.g., \cite{Hallin-Paindaveine-Verdebout-2010} for an appropriate definition of the kurtosis.}.
With a given level $\alpha$, we reject $\H_{0,k}$ if $\hat{\S}_{k,0} \leq \alpha$ and accept $\H_{0,k}$ otherwise.
\end{test}

The interpretation of the test statistic \eqref{eq:random-csv-stat} is quite similar to that of \citeauthor{Choi-Taylor-Tibshirani-2017} (\citeyear{Choi-Taylor-Tibshirani-2017}, Test 2.1): 
The test statistic $\hat{\S}_{k,0}$ compares the relative size of $\hat{\lambda}_k$ ranging between $(\hat{\lambda}_{k+1}, \hat{\lambda}_{k-1})$, and a small value of $\hat{\S}_{k,0}$ implies a large value of $\hat{\lambda}_k$, supporting the alternative hypothesis $\H_{1,k}$ (\citealp{Choi-Taylor-Tibshirani-2017}, p.2598). The test statistic also reflects the selection event in which the first $k-1$ null hypotheses has been rejected in the previous steps of the sequential test.

To establish an asymptotic justification for Test \ref{test:proposed-testing-procedure}, we introduce several assumptions.

\begin{assumption}[Local null structure]\label{ass:local-null}
    We assume that the null hypothesis $\H_{0,k} : \rank(\Sigma_X) \leq k-1$ can be interpreted as follows:
    \begin{itemize}
        \item for each sample size $n$, there exists a population distribution $\mathcal{P}_n$;
        \item under each $\mathcal{P}_n$, the population covariance matrix 
        $\Sigma_X^{(n)}=\E_{\mathcal{P}_n}[X_iX_i^{\top}]$ is deterministic;
        \item among the eigenvalues of $\Sigma_X^{(n)}$, those with indices $j=k,\cdots,p$ satisfy
              \begin{align}
                \lambda_j^{(n)} = \frac{\tau_j^{(n)}}{\sqrt{n}}, \quad 
                \tau_j^{(n)} \to 0, \quad
                (j=k, \cdots,p).
              \end{align}
    \end{itemize}
\end{assumption}

\begin{remarkx}[Interpretation of Assumption~\ref{ass:local-null}.]
This assumption implies that, even when it is not exactly zero, a ``weak'' signal from the population eigenvalues should be negligible in the sense that the corresponding principal component does not have significance in the data. It is natural for analysts to take this view when they would like to obtain ``principal'' components.
\end{remarkx}

\begin{assumption}[Regularity conditions]\label{ass:regularity-conditions}
    For each $n$, we impose the following assumptions.
    The observed $X_i$ ($i=1,\cdots,n$) are i.i.d.\ samples from $X$.
    Without loss of generality, assume $\E_{\mathcal{P}_n}[X_i]=0$.
    Assume that $\Sigma_X^{(n)}=\E_{\mathcal{P}_n}[X_iX_i^{\top}]$ is finite and that
    \begin{align*}
        \E_{\mathcal{P}_n}[\|X_i\|^4] < \infty.
    \end{align*}
    Furthermore, assume that
    \begin{align*}
        \Omega(\Sigma_X^{(n)}) := \E_{\mathcal{P}_n}[\vec(X_iX_i^{\top})\vec(X_iX_i^{\top})^{\top}] - \vec(\Sigma_X^{(n)})\vec(\Sigma_X^{(n)})^{\top}
    \end{align*}
    is finite and positive definite.
\end{assumption}

\begin{assumption}[Population eigenvalue gap]\label{ass:popu-eigen-gap}
    For each $n$, there exists a given constant $c_0 > 0$, which satisfies that, for each $k=1,\cdots,p-1$, under the null hypothesis $\H_{0,k}$,
    \begin{align*}
        \lambda_{j}^{(n)} - \lambda_{j+1}^{(n)} \geq c_0, \quad j=1,\cdots,k-1.
    \end{align*}
\end{assumption} 

\begin{lemmax}[LLN, CLT]\label{lemma:LLN-CLT}
Under Assumptions~\ref{ass:local-null} and \ref{ass:regularity-conditions}, we have, as $n \to \infty$,
\begin{align*}
\hat{\Sigma}_X &\parrow \Sigma_X^{(\infty)}, \\
\sqrt{n}\,\vec(\hat{\Sigma}_X - \Sigma_X^{(n)}) &\darrow \Normal(0,\Omega(\Sigma_X^{(\infty)})).
\end{align*}
\end{lemmax}

\begin{proof}
    We omit the proof of Lemma~\ref{lemma:LLN-CLT}, since it follows immediately from the Cramér-Wold device and the Lindeberg-Feller central limit theorem.
\end{proof}

%%%%%%%%%%%%%%%%%%%%%%%%%%%%%%%%%%%%%%%%%%%%%%%%%%%%%%%%%%%%%%%%%%%%%%%
\begin{lemmax}[Limiting behavior of active eigenvalues]\label{lemma:limiting-nonzero-eigenvalue}
    For $j=1,\cdots,k-1$, it holds that
    \begin{align*}
        \hat{\lambda}_j \parrow \lambda_j^{(\infty)} (\geq c_0 > 0).
    \end{align*}
    Furthermore,
    \begin{align*}
        \sqrt{n} (\hat{\lambda}_j - \lambda_j^{(n)}) =O_p(1).
    \end{align*}
\end{lemmax}

\begin{proof}
    For each $n$, by Lemma~\ref{lemma:LLN-CLT} and Weyl's inequality\footnote{See e.g., \cite{Bhatia-1997}, \cite{Horn-Johnson-2013} and \cite{Yu-Wang-Samworth-2015} for the standard matrix perturbation theory.},
    \begin{align*}
        |\hat{\lambda}_j - \lambda_j^{(n)}| \leq \|\hat{\Sigma}_X - \Sigma_X^{(n)}\|_{op} \parrow 0, \quad j = 1, \cdots, p,
    \end{align*}
    from which the result follows.
\end{proof}
%%%%%%%%%%%%%%%%%%%%%%%%%%%%%%%%%%%%%%%%%%%%%%%%%%%%%%%%%%%%%%%%%%%%%%%
\begin{lemmax}[Sample eigenvalue gap]\label{lemma:sample-eigen-gap}
    For each $k=1,\cdots,p-1$, under the null hypothesis $\H_{0,k}$, for $j=1, \cdots, k-1$, it holds that
    \begin{align*}
        \hat{\lambda}_j > \hat{\lambda}_{j+1} \quad \text{with probability approaching 1}.
    \end{align*}
\end{lemmax}

\begin{proof}
    Follows from Assumption~\ref{ass:popu-eigen-gap} and Lemma \ref{lemma:limiting-nonzero-eigenvalue}.
\end{proof}
%%%%%%%%%%%%%%%%%%%%%%%%%%%%%%%%%%%%%%%%%%%%%%%%%%%%%%%%%%%%%%%%%%%%%%%
\begin{definition}[Null space]
    Under the null hypothesis $\H_{0,k} : \rank(\Sigma_X) \leq k-1$ for each $k=1,\cdots,p-1$, we define the null space of $\Sigma_X$ as
    \begin{align*}
        \N_k = \ker(\Sigma_X^{(\infty)}).
    \end{align*}
\end{definition}

We denote by $P_k$ the projection matrix onto the null space $\N_k$.
By construction, it follows that under the null hypothesis $\H_{0,k}$,
we have $d \equiv \dim(\N_k) = p-k+1$.
%%%%%%%%%%%%%%%%%%%%%%%%%%%%%%%%%%%%%%%%%%%%%%%%%%%%%%%%%%%%%%%%%%%%%%%
\begin{lemmax}[Projected CLT]\label{lemma:projected-CLT}

    For each $k=1,\cdots,p-1$, under the null hypothesis $\H_{0,k}$,
    \begin{align*}
    \sqrt{n} \, \vec \left(P_k (\hat{\Sigma}_X - \Sigma_X^{(n)}) P_k\right) \darrow \vec( G_{\N_k} )       
    \end{align*}
    holds, where $G_{\N_k}$ is a matrix defined on $\N_k$ and
    \begin{align*}
        \vec(G_{\N_k}) \sim \Normal(0, \Xi_{\N_k}), \quad \Xi_{\N_k}=(P_k \otimes P_k) \Omega(\Sigma_X^{(\infty)}) (P_k \otimes P_k)^{\top}.
    \end{align*}
\end{lemmax}

\begin{proof}
     By Lemma~\ref{lemma:LLN-CLT} and the continuous mapping theorem,
    \begin{align*}
        \sqrt{n} \, \vec \left(P_k (\hat{\Sigma}_X - \Sigma_X^{(n)}) P_k\right) 
            &= (P_k \otimes P_k)\sqrt{n} \, \vec\left(\hat{\Sigma}_X - \Sigma_X^{(n)}\right) \\
            &\darrow \Normal\left(0, (P_k \otimes P_k)\Omega(\Sigma_X^{(\infty)})(P_k \otimes P_k)^{\top}\right),
    \end{align*}
    which proves the claim.
\end{proof}

%%%%%%%%%%%%%%%%%%%%%%%%%%%%%%%%%%%%%%%%%%%%%%%%%%%%%%%%%%%%%%%%%%%%%%%
\begin{lemmax}[Limiting behavior of null eigenvalues]\label{lemma:limiting-null-eigenvalue}
    Let the eigenvalues of $G_{\N_k}$, sorted in descending order, be denoted by $\mu_1 > \cdots > \mu_d$.
    Then, for each $k=1,\cdots,p-1$, under the null hypothesis $\H_{0,k}$, it holds that
    \begin{align*}
        (\sqrt{n}\,\hat{\lambda}_k, \cdots, \sqrt{n}\,\hat{\lambda}_p) \darrow (\mu_1, \cdots, \mu_d).
    \end{align*}
\end{lemmax}

\begin{proof}
    Consider $j= k,\cdots,p$.
    Under the null hypothesis $\H_{0,k}$, by the ``local null'' interpretation (Assumption~\ref{ass:local-null}),
    \begin{align*}
         \sqrt{n} \lambda_j^{(n)} = \tau_j^{(n)} \to 0
    \end{align*}
    holds.
    Moreover, by Lemma~\ref{lemma:LLN-CLT} and Weyl's inequality,
    \begin{align*}
        \sqrt{n} (\hat{\lambda}_j - \lambda_j^{(n)}) = O_p(1)
    \end{align*}
    holds.
    Recall also from Lemma~\ref{lemma:projected-CLT} that
    \begin{align*}
        \sqrt{n} \, \vec \left(P_k (\hat{\Sigma}_X - \Sigma_X^{(n)}) P_k\right) \darrow \vec (G_{\N_k}).
    \end{align*}
    Combining these and applying the continuous mapping theorem, we obtain the desired result:
    \begin{align*}
        (\sqrt{n} \, \hat{\lambda}_k, \cdots ,\sqrt{n} \, \hat{\lambda}_p)
            &= \left(\sqrt{n} \lambda_k^{(n)} +  \sqrt{n} (\hat{\lambda}_k - \lambda_k^{(n)}), \cdots, \sqrt{n} \lambda_p^{(n)} +  \sqrt{n} (\hat{\lambda}_p - \lambda_p^{(n)}) \right)\\
            &= \left(\tau_k^{(n)} + \sqrt{n} (\hat{\lambda}_k - \lambda_k^{(n)}),\cdots, \tau_p^{(n)} + \sqrt{n} (\hat{\lambda}_p - \lambda_p^{(n)})\right) \\
            &\darrow  (\mu_1, \cdots, \mu_d).
    \end{align*}
\end{proof}

%%%%%%%%%%%%%%%%%%%%%%%%%%%%%%%%%%%%%%%%%%%%%%%%%%%%%%%%%%%%%%%%%%%%%%%
\begin{assumption}[Wishart-type reference limit experiment]\label{ass:wishart-limit-experiment}
We assume that the limiting joint probability density function of
$(\mu_1,\ldots,\mu_d)$ has the same factorized form as the eigenvalue
density of a Wishart matrix. That is, there exist a constant
$C_d>0$ and a measurable function $h:\mathbb{R}\to[0,\infty)$ such that
\begin{align*}
  f(\mu_1,\ldots,\mu_d)
  = C_d
    \Biggl\{\prod_{i=1}^d h(\mu_i)\Biggr\}
    \prod_{1\le i<j\le d} |\mu_i-\mu_j|.
\end{align*}
\end{assumption}

\begin{remarkx}[Interpretation and scope of Assumption~\ref{ass:wishart-limit-experiment}]
Assumption~\ref{ass:wishart-limit-experiment} is an assumption on the limiting experiment rather than on the exact finite sample distribution of the sample covariance matrix $\hat\Sigma_X$. 
In particular, we do not require $\hat\Sigma_X$ to be Gaussian for each fixed $n$. 
Instead, under the local null hypothesis $\H_{0,k}$, we postulate that the fluctuation of $\hat\Sigma_X$ projected onto the null space $\mathcal{N}_k$ converges, in the sense of asymptotic experiments, to a Wishart-type reference model whose eigenvalue density factorizes as in Assumption~\ref{ass:wishart-limit-experiment}. 
This ``Wishart reference limit'' is precisely what is needed to ensure that the conditioning arguments in Eq.~(B.3) of \cite{Choi-Taylor-Tibshirani-2017}, which rely on the cancellation of the interaction term $\prod_{i<j}|\mu_i-\mu_j|$ and on the separability of the marginal component, remain valid for the limiting eigenvalues.

Moreover, results on the asymptotic behavior of sample covariance eigenvalues and scatter estimators indicate that such a reference experiment arises quite naturally in a broad class of non-Gaussian models with finite fourth moments and elliptically symmetric structure. 
For example, \cite{Waternaux-1976} establishes a multivariate central limit theorem for the vector of sample eigenvalues under finite fourth-moment conditions, showing that their joint limiting law is Gaussian while still inheriting the usual Jacobian structure associated with the eigen decomposition. 
In a complementary direction, \cite{Tyler-1983} considers affine-invariant scatter estimators under elliptical distributions and proves that the suitably normalized estimation error converges to an orthogonally invariant Gaussian symmetric matrix; 
the eigenvalues of such matrices have densities of the form composed of symmetric function and Vandermonde factor, which underpins the Wishart-type structure assumed here. 
Then, perturbation bounds of \cite{Yu-Wang-Samworth-2015} imply that, under our eigenvalue gap condition (Assumption~\ref{ass:popu-eigen-gap}), the eigenspaces and eigenvalues of the projected sample covariance $P_k \hat\Sigma_X P_k$ are asymptotically stable functions of the Gaussian limit in operator norm. 
Taken together, these results justify viewing Assumption~\ref{ass:wishart-limit-experiment} as a convenient Wishart reference limit that captures the local behavior of the null eigenvalues for a wide class of elliptically distributed designs with finite fourth moments, rather than as a strong finite sample Gaussianity requirement.
\end{remarkx}

\begin{lemmax}[Limiting behavior of test statistic]\label{lem:ratio-as-survival}
Under Assumptions~\ref{ass:local-null}-\ref{ass:wishart-limit-experiment}, it holds that, as $n \to \infty$,
\begin{align}\label{eq:reference-limit}
    \hat{\S}_{k,0}
        \darrow \frac{
            \int_{\mu_1}^{\infty} f\!\left(v\,\middle|\,\mu_2,\cdots,\mu_d\right)\,dv
            }{
            \int_{\mu_2}^{\infty} f\!\left(v\,\middle|\,\mu_2,\cdots,\mu_d\right)\,dv
            }
\end{align}
where $f_n(\cdot\,|\,\mu_2,\cdots,\mu_d)$ denotes the conditional density of $\mu_1$ given $(\mu_2,\ldots,\mu_d)$, which on the event $\{\mu_1>\mu_2\}$ takes the form of
\[
f(\mu_1 \mid \mu_2, \dots, \mu_d)\ \propto\ h(\mu_1)\ \prod_{j=2}^d \bigl|\mu_1 - \mu_j\bigr|,
\]
with $h(\cdot)$ denoting the term which is not depending on $(\mu_2,\ldots,\mu_d)$.
\end{lemmax}

\begin{proof}
Recall, from Lemmas~\ref{lemma:LLN-CLT}–\ref{lemma:limiting-null-eigenvalue}, that the active eigenvalues $(\hat\lambda_1,\ldots,\hat\lambda_{k-1})$ converge to deterministic limits and the (scaled) null eigenvalues $(\sqrt n\,\hat\lambda_k,\ldots,\sqrt n\,\hat\lambda_p)$ converges in distribution to $(\mu_1,\ldots,\mu_d)$. In particular, Lemma~\ref{lemma:sample-eigen-gap} ensures gaps between the active sample eigenvalues, so that, uniformly over $u\in[\hat\lambda_{k+1},\hat\lambda_{k-1}]$, the factor $\prod_{j<k}|u-\hat\lambda_j|$ equals $C_k\{1+o_p(1)\}$ for a deterministic $C_k>0$. Hence this contribution cancels between the numerator and denominator of \eqref{eq:random-csv-stat}, and the statistic reduces up to $o_p(1)$ to a ratio of integrals depending only on $(\mu_{1}, \cdots, \mu_{d})$.

Under Assumption~\ref{ass:wishart-limit-experiment}, the joint density of $(\mu_1,\cdots,\mu_d)$ has the Vandermonde factor. Relying on the arguments of \citet[Sec.~2.2 and Appendix~B]{Choi-Taylor-Tibshirani-2017}, conditioning on $(\mu_2,\ldots,\mu_d)$ yields a conditional density for $\mu_1$ of the form $h(\mu_1)\prod_{j=2}^d|\mu_1-\mu_j|$ on $\{\mu_1>\mu_2\}$.

\end{proof}

\begin{propx}[Asymptotically exact size control]\label{prop:Global-size-ctrl}
    Under Assumptions \ref{ass:local-null}-\ref{ass:wishart-limit-experiment}, the proposed test statistic \eqref{eq:random-csv-stat} satisfies that, for each $k=1,\cdots,p-1$, 
    \begin{align}
        \hat{\S}_{k,0} \darrow U(0,1).
        \label{eq:global-size-ctrl}
    \end{align}
    under the null hypothesis $\H_{0,k}$.
\end{propx}

\begin{proof}[Proof]
The argument mirrors the proof of \citet[Theorem~2.1]{Choi-Taylor-Tibshirani-2017}. By Lemma~\ref{lem:ratio-as-survival}, the test statistic admits the representation
\begin{align*}
\hat{\S}_{k,0} \darrow 1-F\!\bigl(\mu_1\,\big|\,\mu_2, \dots, \mu_d\bigr),
\end{align*}
where $F(\cdot\,|\,\mu_2, \dots, \mu_d)$ is the conditional distribution function of the leading null eigenvalue $\mu_1$ given the remainder $(\mu_2, \dots, \mu_d)$. Since a conditional distribution function evaluated at its own random argument is conditionally uniform, the conditional probability integral transform yields that $1-F(\mu_1\,|\,\mu_2, \dots, \mu_d)$ is conditionally $U(0,1)$. 
\end{proof}

The convergence \eqref{eq:global-size-ctrl} yields asymptotically exact type~I error controls for testing \eqref{eq:null-and-alter-hypo}.

%%%%%%%%%%%%%%%%%%%%%%%%%%%%%%%%%%%%%%%%%%%%%%%%%%%%%%%%%%%%%%%%%%%%%%%
\section{Simulation examples}\label{sec:simulation}

We evaluate the finite-sample performance of the proposed method using a data generating process characterized by heavy-tailed, non-Gaussian innovations. 

Let $p$ denote the dimension of $X$. The population covariance matrix $\Sigma_X \in \mathbb{R}^{p \times p}$ is constructed as $\Sigma_X = Q\Lambda Q^\top$,
where $Q$ is a random orthogonal matrix and $\Lambda=\mathrm{diag}(\lambda_1,\ldots,\lambda_p)$. The eigenvalues follow a spiked model:
\begin{align*}
    \lambda_i &= 10 - 5\left(\frac{i-1}{k_0-1}\right), \quad i=1,\ldots,k_0,\\
    \lambda_i &= 1, \quad i=k_0+1,\ldots,p,
\end{align*}
where $k_0$ denotes the true rank. 

To introduce non-Gaussianity, we generate data $X_i = \Sigma_X^{1/2} Z_i$, where the vector $Z_i = (Z_{i,1}, \dots, Z_{i,p})^\top$ consists of independent standardized Student-$t$ random variables. Specifically,
\[
    Z_{i,j} = \sqrt{\frac{\nu-2}{\nu}} \, T_{i,j}, \quad T_{i,j} \sim t_\nu,
\]
where $t_\nu$ denotes the Student-$t$ distribution with $\nu$ degrees of freedom. We set $\nu = 5$. This choice ensures that $\E[Z_{i,j}]=0$, $\Var(Z_{i,j})=1$, and the fourth moment exists. 
This design allows us to assess whether the proposed test maintains type~I error control even when the Gaussian assumption is violated.

Table~\ref{tab:simulation_results_p_is_10} reports the results for $(p,k_0)=(10,3),(10,4),(10,5),(10,6),(10,7)$, 
%Table~\ref{tab:simulation_results_p_is_50} for $(p,k_0)=(50,5),(50,15),(50,25),(50,35),(50,45)$, and 
%Table~\ref{tab:simulation_results_p_is_100} for $(p,k_0)=(100,10),(100,30),(100,50),(100,70),(100,90)$, 
with sample sizes $n\in\{100,1000,10000,100000\}$, and nominal sizes $\alpha\in\{0.01,0.05,0.10\}$. 
Each configuration is based on 1000 replications. In this setting, over-estimation corresponds to a type~I error (spuriously adding components), whereas under-estimation corresponds to a type~II error (failing to detect existing components). Consistent with our theory, the empirical size approaches the nominal level as $n$ increases, indicating effective size control in large samples. Although a formal proof is not yet available, the simulation evidence further suggests that, under local alternatives, the power tends to one asymptotically.

\input{table-p-is-10}
%\input{table-p-is-50}
%\input{table-p-is-100}

%%%%%%%%%%%%%%%%%%%%%%%%%%%%%%%%%%%%%%%%%%%%%%%%%%%%%%%%%%%%%%%%%%%%%%%
%\section{Empirical illustrations}\label{sec:empirical}

%\noindent
%\textcolor{purple}{TO BE UPDATED}

%%%%%%%%%%%%%%%%%%%%%%%%%%%%%%%%%%%%%%%%%%%%%%%%%%%%%%%%%%%%%%%%%%%%%%%
\section{Conclusions}\label{sec:conclusions}
This paper revisits the classical problem of determining the number of principal components through the lens of selective inference. We formalized the rank selection task as a sequence of nested tests and proposed a statistic that depends on the eigenvalues of the sample covariance matrix. In contrast to fixed design approaches, our analysis treats the design as random and avoids parametric distributional assumptions, relying instead on asymptotic theory and a locally defined null. Within this framework, we established asymptotically exact type~I error controls in the sequential testing procedure via a Wishart reference limit experiment that justifies the relevant conditioning arguments. Simulation results support the theoretical predictions: empirical sizes approach nominal levels as the sample size increases, and the procedure behaves robustly in non-Gaussian settings.

Several directions merit further investigation. First and foremost, while the simulation results indicate that the power under local alternatives increases with the sample size, a formal proof that the power converges to one asymptotically remains to be established. Providing such a result would sharpen our understanding of the method’s efficiency and deliver stronger guarantees for asymptotic optimality under local alternatives. 

Taken together, our results contribute a conceptually simple and practically implementable rank selection procedure grounded in selective inference. 
Our framework provides a useful bridge between classical PCA practice in econometrics and modern inferential tools, and we view the formal characterization of power under local alternatives as a future task.

%%%%%%%%%%%%%%%%%%%%%%%%%%%%%%%%%%%%%%%%%%%%%%%%%%%%%%%%%%%%%%%%%%%%%%%
\bibliographystyle{apalike} 
\bibliography{refs}

%%%%%%%%%%%%%%%%%%%%%%%%%%%%%%%%%%%%%%%%%%%%%%%%%%%%%%%%%%%%%%%%%%%%%%%
%\newpage
%----------------------------------------------------------------%
%                      Supplement                                %
%----------------------------------------------------------------%
%\input{preamble-appendix}

%\begin{center}
%{\Large\bf Online Appendix for\\ ``Principal component regression in econometrics:\\a selective inference perspective''}\\
%\vspace{1cm}
%\large{Yasuyuki Matsumura$^{\clubsuit}$ and Chisato Tachibana$^{\spadesuit}$}\\
%\vspace{1cm}
%$^{\clubsuit}$Graduate School of Economics, Kyoto University\\
%$^{\spadesuit}$Graduate School of Economics, University of Tokyo\\
%\vspace{1cm}
%{\today}\\
%\vspace{1cm}
%\end{center}

\end{document}

%% file: preamble-main.tex
\usepackage{mathpazo}
\usepackage{lineno}
\usepackage{amssymb}
\usepackage{graphicx}
\usepackage{amsfonts}
\usepackage{amsmath}
\usepackage{amsthm}
\usepackage{mathtools}
\usepackage{ascmac}
\usepackage{bm}
\usepackage{color}
\usepackage{pgfplots}
\pgfplotsset{compat=1.16}
\usepackage{here}
\usepackage{booktabs} % 表の罫線をきれいにするために必要
\usepackage{multirow}
\usepackage[hang,small,bf]{caption}
\usepackage[subrefformat=parens]{subcaption}
\usepackage[colorlinks=true, citecolor=teal, linkcolor=teal, urlcolor=teal]{hyperref}
\usepackage[authoryear]{natbib}
\usepackage{threeparttable}
\usepackage{comment}
\usepackage{stackengine}
\usepackage{tikz}
\usetikzlibrary{positioning, calc}
\usetikzlibrary{arrows.meta}
\usetikzlibrary{bending}

\def\E#1{\mathbb{E}\left[#1\right]}

\def\ker{\mathrm{ker}}
\def\E{\mathbb{E}}

\def\H{\mathbb{H}}
\def\S{\mathbb{S}}
\def\N{\mathcal{N}}
\def\rank{\mathrm{rank}}
\def\Normal{\mathrm{Normal}}
\def\darrow{\xrightarrow{d}}
\def\parrow{\xrightarrow{p}}
\def\dim{\mathrm{dim}}
\def\vec{\mathrm{vec}}
\def\lim{\mathrm{lim}}

\def\Var{\mathrm{Var}}

\def\hat{\widehat}

\numberwithin{equation}{section}

\newtheorem{theoremx}{Theorem}

\newtheorem{assumption}{Assumption}
\newtheorem{test}{Test}

\newtheorem{definition}[theoremx]{Definition}

\newtheorem{lemmax}{Lemma}

\newtheorem{propx}{Proposition}

\theoremstyle{definition}
\newtheorem{remarkx}{Remark}

\makeatletter
\renewenvironment{proof}[1][\proofname]{%
  \par\pushQED{\qed}\normalfont%
  \topsep6\p@\@plus6\p@\relax
  \trivlist\item[\hskip\labelsep\bfseries#1\@addpunct{.}]%
  \ignorespaces
}{%
  \popQED\endtrivlist\@endpefalse
}
\makeatother

\makeatletter
\renewcommand*{\@fnsymbol}[1]{%
  \ensuremath{%
    \ifcase#1\or
      \flat\or          % 1
      *\or              % 2
      \dagger\or        % 3
      \ddagger\or       % 4
      \mathsection\or   % 5
      \mathparagraph\or % 6
      \|\or             % 7
      **\or             % 8
      \dagger\dagger\or % 9
      \ddagger\ddagger  % 10
    \else
      \@ctrerr
    \fi
  }%
}
\makeatother

%% file: table-p-is-10.tex
\begin{table}[htbp]
\centering
\caption{Simulation results with $p=10$ (1000 iterations)}
\label{tab:simulation_results_p_is_10}

\resizebox{\textwidth}{!}{%
\begin{tabular}{cl rrr rrr rrr}
\toprule
\multirow{2}{*}{True Rank} & \multirow{2}{*}{Sample Size}
& \multicolumn{3}{c}{Results ($\alpha=0.01$)}
& \multicolumn{3}{c}{Results ($\alpha=0.05$)}
& \multicolumn{3}{c}{Results ($\alpha=0.10$)} \\
\cmidrule(lr){3-5}\cmidrule(lr){6-8}\cmidrule(lr){9-11}
& & $\hat{k} < k_0$ & $\hat{k} = k_0$ & $\hat{k} > k_0$ & $\hat{k} < k_0$ & $\hat{k} = k_0$ & $\hat{k} > k_0$ & $\hat{k} < k_0$ & $\hat{k} = k_0$ & $\hat{k} > k_0$ \\
\midrule

\multirow{4}{*}{$k_0=3$}
& $n=100$     & 55 & 904 & 41  & 26 & 827 & 147  & 17 & 750 & 233 \\
& $n=1,000$   &  0 & 961 & 39  &  0 & 892 & 108  &  0 & 830 & 170 \\
& $n=10,000$  &  0 & 980 & 20  &  0 & 920 &  80  &  0 & 866 & 134 \\
& $n=100,000$ &  0 & 981 & 19  &  0 & 932 &  68  &  0 & 884 & 116 \\
\midrule

\multirow{4}{*}{$k_0=4$}
& $n=100$     & 218 & 749 & 33 & 118 & 749 & 133 & 77 & 691 & 232 \\
& $n=1,000$   &   2 & 967 & 31 &   1 & 900 &  99 &  0 & 837 & 163 \\
& $n=10,000$  &   0 & 985 & 15 &   0 & 935 &  65 &  0 & 881 & 119 \\
& $n=100,000$ &   0 & 991 &  9 &   0 & 940 &  60 &  0 & 890 & 110 \\
\midrule

\multirow{4}{*}{$k_0=5$}
& $n=100$     & 528 & 450 & 22 & 308 & 594 &  98 & 225 & 587 & 188 \\
& $n=1,000$   &   9 & 970 & 21 &   6 & 902 &  92 &   2 & 850 & 148 \\
& $n=10,000$  &   0 & 987 & 13 &   0 & 947 &  53 &   0 & 885 & 115 \\
& $n=100,000$ &   0 & 989 & 11 &   0 & 953 &  47 &   0 & 897 & 103 \\
\midrule

\multirow{4}{*}{$k_0=6$}
& $n=100$     & 857 & 140 &  3 & 596 & 348 &  56 & 458 & 412 & 130 \\
& $n=1,000$   &  35 & 948 & 17 &  23 & 892 &  85 &  20 & 833 & 147 \\
& $n=10,000$  &   0 & 990 & 10 &   0 & 951 &  49 &   0 & 897 & 103 \\
& $n=100,000$ &   0 & 994 &  6 &   0 & 953 &  47 &   0 & 909 &  91 \\
\midrule

\multirow{4}{*}{$k_0=7$}
& $n=100$     & 991 &   9 &  0 & 889 &  96 &  15 & 766 & 170 &  64 \\
& $n=1,000$   & 131 & 845 & 24 &  69 & 860 &  71 &  49 & 816 & 135 \\
& $n=10,000$  &   0 & 994 &  6 &   0 & 963 &  37 &   0 & 892 & 108 \\
& $n=100,000$ &   0 & 994 &  6 &   0 & 952 &  48 &   0 & 911 &  89 \\
\bottomrule
\end{tabular}%
}

\begin{minipage}{0.95\linewidth}
\vspace{2mm}
\footnotesize
\textbf{Note:} $\hat{k}$ denotes the estimated rank. Each number in the cells represents the count of under-estimated ($\hat{k} < k_0$), correctly estimated ($\hat{k} = k_0$), and over-estimated ($\hat{k} > k_0$) ranks out of 1,000 iterations.
\end{minipage}

\end{table}